\documentclass[journal]{IEEEtran}
\usepackage{cite}
\usepackage{amsmath,amssymb,amsfonts}
\usepackage{algorithmic}
\usepackage{graphicx}
\usepackage{textcomp}
\usepackage{xcolor}
\usepackage{url}
\def\BibTeX{{\rm B\kern-.05em{\sc i\kern-.025em b}\kern-.08em
    T\kern-.1667em\lower.7ex\hbox{E}\kern-.125emX}}

\usepackage{subcaption}    
\usepackage[ruled,vlined]{algorithm2e}

\newtheorem{theorem}{Theorem}[section]
\newtheorem{corollary}{Corollary}[section]

\newtheorem{lemma}{Lemma}[section]

\newtheorem{definition}{Definition}[section]
\def\eop{\hfill$\Box$}
\def\proof{{\em Proof: }}

\def\R{{\mathbb R}}
\def\C{{\mathbb C}}

\def\A{{\cal A}}
\def\B{{\cal B}}
\def\M{{\cal M}}
\def\S{{\cal S}}

\begin{document}

\title{Synchronization in dynamical systems coupled via multiple directed networks}

\author{\IEEEauthorblockN{Chai Wah Wu}\\
\IEEEauthorblockA{\textit{IBM T.J. Watson Research Center} \\
Yorktown Heights, NY 10598, USA \\
cwwu@us.ibm.com\\March 24, 2021}
}

\bibliographystyle{ieeetr}

\maketitle

\begin{abstract}
We study synchronization and consensus in a group of dynamical systems coupled via multiple directed networks.  We show that even though the coupling in a single network may not be sufficient to synchronize the systems, combination of multiple networks can contribute to synchronization.  We illustrate how the effectiveness of a collection of networks to synchronize the coupled systems depends on the graph topology.  In particular, we show that if the graph sum is a directed graph whose reversal contains a spanning directed tree, then  the network synchronizes if the coupling is strong enough. This is intuitive as there is a root node that influence every other node via a set of edges where each edge in the set is in one of the networks.
\end{abstract}

\begin{IEEEkeywords}
nonlinear dynamical systems, graph theory
\end{IEEEkeywords}

\section{Introduction}

There has been much research activity in studying the synchronization and consensus behavior in networks of coupled dynamical systems 
\cite{synch_special_issue:1997,chaos:synch_special_issue:1997,synch_special_issue_ijbc:2000,chaos:synch_special_issue:2003,wu:synch_book:2007,chaos:synch_special_issue:2008,kocarev:synch:2013,Baldi2020,Rahnama2020}.
The coupled systems model various types of networks both man-made and naturally occurring, such as social networks, transportation networks, genetic networks and neural networks.   In many cases the individual systems or {\em agents} are connected via multiple networks.  For instance, automobiles are connected via a transportation network, but they can also connected via cell phones (or other vehicle-to-vehicle technologies such as DSRC \cite{Abboud2016}) in a communication network.  Similarly, people can be connected via multiple social networks such as Facebook and LinkedIn. The Erd\H{o}s-Bacon number is the sum of the distance in the actor and in the mathematical collaboration network.

Prior work on 
systems coupled via multiple networks include \cite{wu:synch-time-varying-delay:2005}
where synchronization is studied in continuous-time dynamical systems coupled via $2$ networks with one of the networks coupling delayed state variables and conditions were provided under which the system is synchronizing.
Refs. \cite{Sorrentino:2012,Gambuzza2015,Genio2016,SevillaEscoboza2016,Leyva2018,Majhi2019,Tang2019,Blaha2019,Bera2019,DellaRossa2020,Kumar2020,BurbanoL.2020} also studied synchronization in continuous-time dynamical systems coupled via multiple networks where a local stability approach numerically computing Lyapunov exponents or a mean field approximation is used. In some of these works, networks are studied whose Laplacian matrices  are simultaneously diagonalizable. This paper differs from prior work in that 1) we use the Lyapunov function approach which provides a guaranteed synchronization criteria, 2) consider simultaneously triangularizable matrices which is a superset of the simultaneously diagonalizable matrices, 3) consider directed networks 4) and consider extensions to discrete-time systems.

In this paper we study synchronization and consensus in dynamical systems coupled under multiple networks and study how these multiple networks can work together to reach synchronization even though each individual network might not be able to effect synchronization.  
We consider the case of continuous time systems in Section \ref{sec:continuous}.  
The case of discrete time systems will be studied in Section \ref{sec:discrete}.

\section{Mathematical preliminaries}
We define $e_i$ as the $i$-th unit vector and $E_i$ as the diagonal square matrix such that all entries are $0$ except that the $i$-th entry on the diagonal is $1$
and define $e$ as the vector of all $1$'s and $I$ as the identity matrix.  For a Hermitian matrix, the eigenvalues are listed as $\lambda_1\leq \lambda_2\cdots \leq\lambda_n$.
We use $A\succ 0$ and $A \succeq 0$ to denote that the matrix $A$ is positive definite and positive semidefinite respectively\footnote{For non-Hermitian matrices, we say $A$ is positive (semi)definite if $A+A^H$ is positive (semi)definite.}. 
The following simple results generalize the results in \cite{wu:chua_kronecker_product_synch_95} and in \cite[Lemma 3.1]{wu:synch_book:2002} and prove to be useful when studying coupling via multiple networks:

\begin{lemma}\label{lem:kron-ev}
Let $\M = \{A_i\}$ be a set of $n$ by $n$ simultaneously triangularizable matrices, i.e. there exists a nonsingular matrix $P$ which transforms these matrices simultaneously into upper triangular form, i.e. $\exists P\forall i,
PA_iP^{-1}= J_i$ is upper triangular with corresponding eigenvalues $\lambda_{ij}$ on the diagonal of $J_i$.  Let $\{B_i\}$ be a set of $m$ by $m$ matrices.  Then the $nm$ by $nm$ matrix
$C = \sum_i A_i\otimes B_i$ has eigenvalues $\gamma_{jk}$ where $\gamma_{jk}$ are the eigenvalues of $\sum_i  \lambda_{ij}B_i$ for each $j$.
\end{lemma}
\begin{proof}
The matrix  $(P\otimes I)C(P^{-1}\otimes I) = \sum_i (J_i\otimes B_i)$ is block upper triangular with diagonal blocks equal to $\sum_i  \lambda_{ij}B_i$.  Thus the eigenvalues are the same as the eigenvalues of $\sum_i  \lambda_{ij}B_i$ for $1\leq j\leq n$.\footnote{Since all $G_i$ share the same set of unit length eigenvectors $v_j$, we denote $\lambda_{ij}$ as the eigenvalue of $B_i$ corresponding to the eigenvector $v_j$. This convention will used throughout.}
\eop
\end{proof}

It is clear that if $\M$ satisfies the condition of Lemma \ref{lem:kron-ev}, so does $\M \cup \{\alpha I\}$ where $\alpha \in \C$.

\begin{corollary}
[\cite{wu:chua_kronecker_product_synch_95}]
The eigenvalues of $I\otimes B_1 + A\otimes B_2$ are equal to the eigenvalues of $B_1+\lambda_i B_2$ for each $i$, where $\lambda_i$ are the eigenvalues of $A$.
\end{corollary}

\begin{corollary}
  Let $p_i$ be polynomials and $A$ be a square matrix with eigenvalues $\lambda_j$.  Then $C = \sum_i p_i(A)\otimes B_i$ has eigenvalues $\gamma_{jk}$ where $\gamma_{jk}$ are the eigenvalues of $\sum_i p_i(\lambda_j)B_i$ for each $j$.
\end{corollary}
\begin{proof} The matrix $A$ is similar to an upper triangular matrix $J$ (e.g. the Jordan normal form). If $J$ is in upper triangular form, so is $J^2$ and it is clear from the spectral mapping theorem that if $J$ has eigenvalues $\lambda_i$, then $p(J)$ has eigenvalues $p(\lambda_i)$.   This implies that the set of matrices $\{p_1(A),p_2(A),\cdots, p_n(A)\}$ where $A$ is a fixed matrix and $p_i$ are polynomials can be transformed to an upper triangular form via the same matrix $P$ and the result follows from Lemma \ref{lem:kron-ev}. 
\eop
\end{proof}

It is well known that the set of simultaneously triangularizable matrices corresponds to the set of commuting matrices \cite{horn-johnson:matrix_analysis:1985}.
Examples of commuting sets of matrices include the set of $n$ by $n$ circulant matrices.  These matrices are diagonalizable with eigenvectors equal to the columns of the Discrete Fourier Transform matrix.  
As noted before, the identity matrix $I$ (and its scalar multiple) can always be added to such sets since the identity matrix commutes with any matrix.

In our study, $A_i$ will be Laplacian matrices of directed graphs\footnote{A Laplacian matrix of a directed graph is defined as $D-A$ where $D$ is the diagonal matrix of vertex outdegrees and $A$ is the adjacency matrix.} and $A_ie = 0$ for all $i$.  This means that for all vector $v$, $e\otimes v$ is an eigenvector of $A_i\otimes B_i$ corresponding to eigenvalue $0$, i.e. the eigenvalue $0$ has multiplicity at least $m$.\footnote{Recall that $m$ is the order of the matrices $B_i$.}
If in addition the graph ${\mathcal G}_i$ corresponding to $A_i$ is connected, then $0$ is an eigenvalue of $A_i$ with multiplicity $1$.  If $B_i$ is nonsingular, then $A_i\otimes B_i$ has a 0 eigenvalue of multiplicity $m$.  

We ask: when does $\sum_i A_i\otimes B_i$ have a zero eigenvalue of multiplicity $m$?  The answer to this question will be useful later on in our study of synchronization in network of systems coupled via multiple networks.  The following results give some conditions under which this question can be answered.  

\begin{theorem}
If $B$ is nonsingular, and $\sum_i A_i$ has a zero eigenvalue of multiplicity 1, then $\sum_i A_i\otimes B$ has a zero eigenvalue of multiplicity $m$.
\end{theorem}
\begin{proof} The eigenvalues of $\sum_i A_i\otimes B$  are the eigenvalues of $\lambda_iB$ where $\lambda_i$ are the eigenvalues of $A$.
\eop
\end{proof}

\begin{definition}
The {\em graph sum} of (weighted) graphs $G_1,G_2\cdots, G_r$ with the same vertex set is a graph with adjacency matrix equal to the sum of the adjacency matrices of $G_i$.
\end{definition}
For example, Fig. \ref{fig:cyclic2} shows two directed graphs with the same vertex set and edges labeled `A' and `B' resp.  Under the canonical ordering of the nodes, the Laplacian matrices of these graphs are circulant and thus commuting.  Even though each of the graph is not connected, the graph sum is a connected graph and between every pair of nodes there is a directed path between them using either `A' or `B' edges.

\begin{figure}[htbp]
\centerline{\includegraphics[width=3.5in]{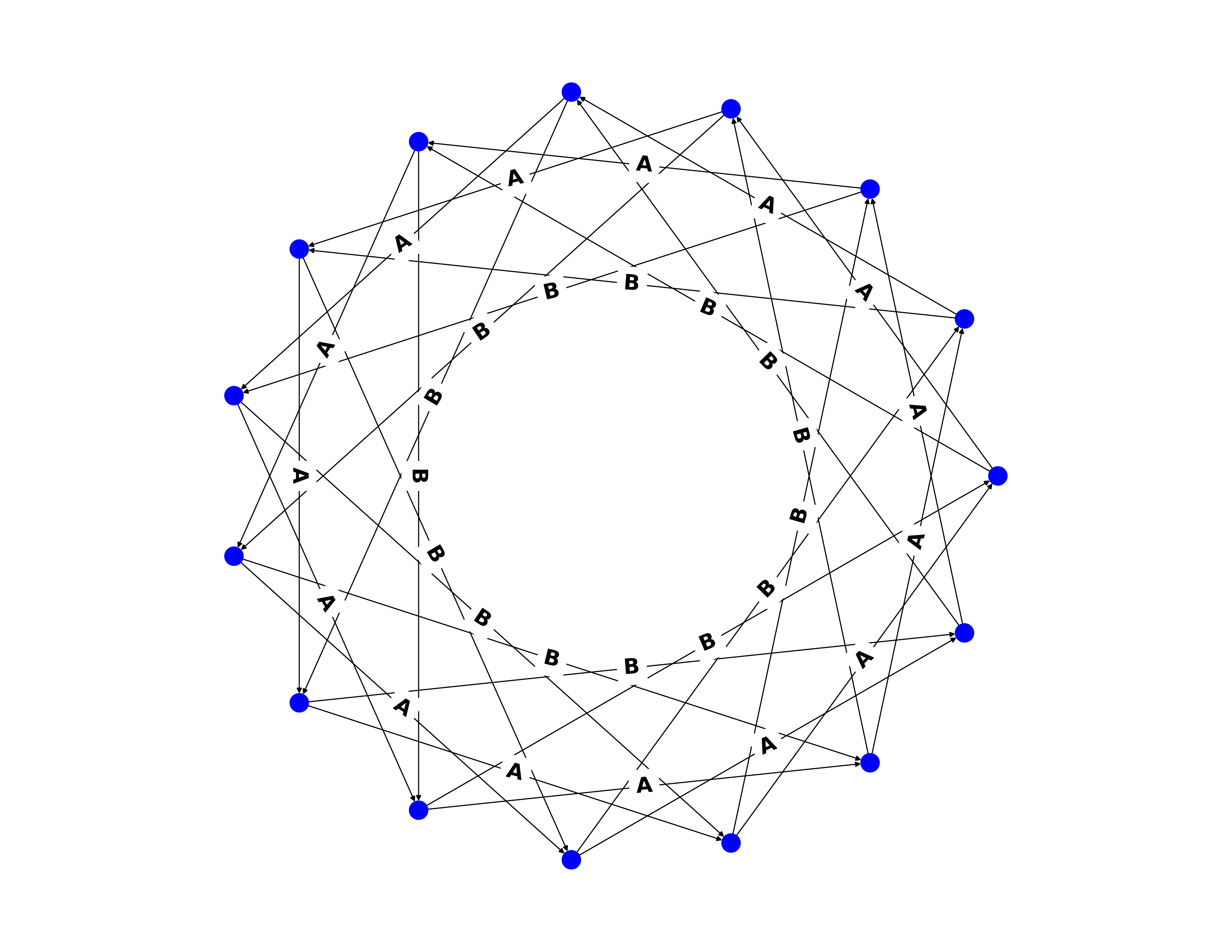}}
\caption{Two directed graphs with the same vertex set superimposed.  The edges of the two graphs are labeled `A' and `B' respectively.}

\label{fig:cyclic2}
\end{figure}

Note that the graph sum depends on the ordering of the vertices and is thus not invariant under graph isomorphism, i.e. we consider labeled graphs.
If $A_i$ are the Laplacian matrices of graphs, then $\sum_i A_i$ is the Laplacian matrix of the graph sum of these graphs and $\sum_i A_i$ having a zero eigenvalue of multiplicity 1 is equivalent to the reversal\footnote{i.e. the directed graph obtained by reversing the orientation of all edges.} of the graph sum containing a spanning directed tree \cite{wu:reducible:2005}.  Note that for directed graphs, the Laplacian matrix is not necessarily symmetric.

\begin{lemma} \label{lem:graph1}
Let $L_i$ be the Laplacian matrices of graphs ${\mathcal G}_i$ with the same vertex set such that the graph sum is a directed graph whose reversal contains a spanning directed tree.  If $L_i$ are simultaneously triangularizable with eigenvalues $\lambda_{ij}$, then for exactly one fixed $j$ is it true that $\lambda_{ij} = 0$ for all $i$.  In fact, the index $j$ corresponds to the zero eigenvector $e$.
\end{lemma}

\begin{proof}
The Laplacian matrix of the graph sum is $\sum_i L_i$.  By the simultaneously triangularizable property, its eigenvalues are $\sum_i \lambda_{ij}$ for each $j$.
By \cite{wu:reducible:2005}, its Laplacian matrix has a $0$ eigenvalue of multiplicity $1$ and the conclusion follows.\eop
\end{proof}

\begin{theorem} \label{thm:spanning_tree}
Let $\{L_i\}$ be a set of $n$ by $n$ simultaneously triangularizable Laplacian matrices, with corresponding eigenvalues $\lambda_{ij}$ that are all nonnegative.  Let $\{B_i\}$ be a set of $m$ by $m$ nonsingular matrices such that any convex combination of $B_i$ is nonsingular.  If the graph sum of the graph corresponding to $\{L_i\}$ is a directed graph whose reversal contains a spanning directed tree,  then the $nm$ by $nm$ matrix
$C = \sum_i L_i\otimes B_i$ has a zero eigenvalue of multiplicity $m$. 
\end{theorem}
\begin{proof} By Lemma \ref{lem:kron-ev},
the eigenvalues of $C$ are given by the eigenvalues $\sum_i \lambda_{ij}B_i$ for each $j$. By Lemma \ref{lem:graph1} only for one $j$ is it true that $\lambda_{ij} = 0$ for all $i$, in which case $\sum_i \lambda_{ij}B_i = 0$ and we have $m$ zero eigenvalues.  For all other $j$ there exists nonzero $\lambda_{ij}$'s.  Since $\lambda_{ij}\geq 0$, this can be rescaled to a convex combination of $B_i$ which is nonsingular by hypothesis. \eop
\end{proof}

Examples of nonsingular matrices $B_i$ for which convex combinations are nonsingular include diagonal positive definite matrices or strictly diagonally dominant matrices whose eigenvalues have positive real parts.

\section{Continuous-time systems} \label{sec:continuous}
For the continuous-time case, consider the state equation:

\begin{equation} \dot{x} = \left(\begin{array}{c} f(x_1,t)\\ \vdots \\ f(x_n,t)\end{array}
\right) - \sum_{k=1}^r(G_k(t)\otimes D_k(t)) x
\label{eqn:main}
\end{equation}
where $x = (x_1,\dots,x_n)$ and $G_k(t)$ and $D_k(t)$ are square matrices for each $k$ and $t$.   There are $n$ {\em identical} systems described  by $\dot{x_i} = f(x_i,t)$ and each $x_i$ is the $m$-dimensional state vector of the $i$-th system.  The coupling matrices $G_k(t)$ have zero row sums and describe
$r$ different coupling networks coupling these systems together, i.e. the $ij$-th entry of the matrix $G_k(t)$ is nonzero if there is a
coupling from the $j$-th system to the $i$-th system via the $k$-th network.  The
matrix $D_k(t)$ describes the coupling between a pair of systems in the $k$-th network and is the same between any pair of systems in the $k$-th network. We are interested in the case where the $D_k$ are not all equal, as otherwise it reverts to the single network case with the network being the graph sum.
We say the system is {\em globally synchronizing} if $\|x_i-x_j\|\rightarrow 0$
as $t\rightarrow\infty$ for all $i$,$j$ and all initial conditions.  
  One important special case is when $D_k = E_k$, which corresponds to the case where the $k$-th network only couples the $k$-th component of the state vectors and the internetwork connections occur solely due to the dynamics of the individual systems.

\begin{definition}
$\cal W$ is the class of symmetric matrices that have zero row sums and
nonpositive off-diagonal elements.
\end{definition}

The following result provides sufficient conditions under which the system in Eq. (\ref{eqn:main}) is globally synchronizing.
and follows directly from the synchronization theorems in \cite{wu:chua_synch_arrays_95,wu:perturbation:2003,wu:synch-directed:2005}.   
\begin{theorem}\label{thm:synch}
Let $P(t)$ be a time-varying matrix and $V\succ 0$ be a symmetric matrix such that $(y-z)^TV(f(y,t)-P(t)y-f(z,t)+P(t)z) \leq -c\|y-z\|^2$ for some
$c>0$.  Then the network in Eq. (\ref{eqn:main}) is globally synchronizing if there exists an irreducible matrix $U\in \cal W$
such that for all $t$
\begin{equation} (U\otimes V)\left(\sum_{k=1}^r G_k(t)\otimes D_k(t) - I \otimes P(t)\right) \succeq 0 \label{eqn:cond} \end{equation}
\end{theorem}

We next look at various special cases of Theorem \ref{thm:synch}.

\subsection{The case $P = \xi I$}
Note that the matrix $P(t)$ can be thought of as a stabilizing (time-varying) negative linear state feedback such that $\dot{x}=f(x,t)-P(t)x$ is asymptotically stable.
Using Theorems \ref{thm:spanning_tree},  \ref{thm:synch} and choosing $P = \xi I$, we get:
\begin{corollary} \label{cor:spanning_tree_synch}
Let $\{L_i\}$ be a set of commuting normal Laplacian matrices where the graph sum's reversal contains a spanning directed tree. If $D_k$ are normal, $f$ is Lipschitz continuous with Lipschitz constant $c$, $G_k = \xi L_k$ and matrices in conv($D_k$) have eigenvalues with positive real part, then Eq. (\ref{eqn:main}) is synchronizing if $\xi\geq \frac{c}{Re(\lambda_L)Re(\lambda_D)}$ where conv($D_k$) is the convex hull of $D_k$, $\lambda_L$ is the nonzero eigenvalue of $L_i$ with minimal real part and $\lambda_D$ is the eigenvalue of matrices in conv($D_k$) with minimal real part. 
\end{corollary}
This result can be shown by noting that $e\otimes v$ are exactly the eigenvectors corresponding to zero eigenvalues of both $U\otimes V$ and $\sum_{k=1}^r G_k(t)\otimes D_k(t)$. Note that by hypothesis $Re(\lambda_D) > 0$ as the convex hull is compact.
Corollary \ref{cor:spanning_tree_synch} is intuitive in the sense that if the graph sum's reversal contains a spanning directed tree, there is at least one node that influences every other node via edges where each edge is in at least one of the networks. Even though each of the network might not be connected by themselves, synchronization can occur via the graph sum if the coupling is strong enough.

The converse is true as well in the following sense. If $G_k$ and $D_k$ do not depend on $t$ and the graph sum does not contain a spanning directed tree, then there exists two nodes in the graph sum that do not influence each other and for general dynamical systems, especially chaotic systems, it is not possible for the systems at these 2 nodes to synchronize to each other.

\subsection{$\dot{x} = f(x,t) - \xi\sum_k D_k(t)$ is asymptotically stable}
Consider the case where $\dot{x} = f(x,t) - \xi\sum_k D_k(t)$ is asymptotically stable for some $\xi > 0$.  In this case, we can choose $P(t)=\xi\sum_k D_k(t)$.  Some examples are when $P(t)$ is a diagonal matrix\footnote{In several prior studies $P$ is chosen to be the identity matrix.} for all $t$ and $D_i(t)=c_i(t)E_i$.  
In this case Eq. (\ref{eqn:cond}) simplifies to
$ (U\otimes V)\left(\sum_{k=1}^r (G_k(t)-\xi I)\otimes D_k(t)\right) \succeq 0 $.
If in addition $VD_k\succeq 0$, such as when $V$ is a diagonal matrix and positive definite and $D_k$ are diagonal and positive semidefinite matrices, then Eq. (\ref{eqn:cond}) is satisfied if
$ U(G_k(t)-\xi I)\succeq 0 $
for all $k=1,\cdots,r$. 

Thus we have shown the following:
\begin{theorem}\label{thm:synch2}
Let $V\succ 0$ be some symmetric matrix such that $(y-z)^TV(f(y,t)-\xi\sum_k D_k y-f(z,t)+\xi\sum_k D_k z) \leq -c\|y-z\|^2$ for some
$c,\xi >0$.  Suppose that $VD_k\succeq 0$ for all $k$. Then the network in Eq. (\ref{eqn:main}) is globally synchronizing if there exists an irreducible matrix $U\in \cal W$
such that 
$U(G_k(t)-\xi I)\succeq 0 $
for all $t$ and all $k=1,\cdots,r$.
\end{theorem}

For $U=L_K$, the Laplacian matrix of the complete graph,
the proof of Theorem 3 in \cite{wu:perturbation:2003} can be used to show the following:
\begin{theorem}\label{thm:synch2b}
Let $V\succ 0$ be some symmetric matrix such that $(y-z)^TV(f(y,t)-\xi\sum_k D_k y-f(z,t)+\xi\sum_k D_k z) \leq -c\|y-z\|^2$ for some
$c, \xi >0$.  Suppose that $VD_k\succeq 0$ and $G_k(t)$ has zero row sums and zero column sums such that $G_k(t)+G_k^T(t)$ has a simple zero eigenvalue. Then the network in Eq. (\ref{eqn:main}) is globally synchronizing if 
$\lambda_2\left(\frac{1}{2}\left(G_k(t)+G^T_k(t)\right)\right) \geq \xi$ for all $t$ and all $k = 1,\cdots , r$.
\end{theorem}
\proof $U = L_K = nI-J$, where $J$ is the $n$ by $n$ matrix of all $1$'s.  Then $U(G_k(t)-I) = nG_k(t)-L_K$. Now $nG_k(t)-L_K \geq 0$ if all eigenvalues of
$\frac{1}{2}n\left(G_k(t)+G^T_k(t)\right)-L_K$ are nonnegative.
For the vector $e$, $\left(\frac{1}{2}n\left(G_k(t)+G^T_k(t)\right)-L_K\right)e = 0$.  For a vector $x$ orthogonal to $e$,
$L_K x = nx$ and $x^T\left(\frac{1}{2}\left(G_k(t)+G^T_k(t)\right)\right)x\geq x^Tx$
and thus the proof is complete.\eop

The interpretation here is that even though the coupling matrix $D_k$ corresponding to an individual network $G_k$ is not sufficient to drive $f$ to stability, the combination of all $D_k$ is sufficient stabilizing feedback.  The consequence is that even though a single coupled network $G_k$ with coupling matrix $D_k$ is not sufficient to achieve synchronization, the combination of multiple networks is.

Next, define $\Xi$ as the set of real numbers $\xi$ such that there exists an irreducible matrix $U\in {\cal W}$ satisfying
$U(G_k(t)-\xi I)\succeq 0$ for all $k$ and $t$.  Let $\xi_M$ be the supremum of $\Xi$.

\begin{corollary}\label{cor:synch2}
Let $V\succ 0$ be some symmetric matrix such that $(y-z)^TV(f(y,t)-\xi_M\sum_k D_k y-f(z,t)+\xi_M\sum_k D_k z) \leq -c\|y-z\|^2$ for some
$c>0$ and that $VD_k\succeq 0$ for all $k$. Then Eq. (\ref{eqn:main}) is globally synchronizing.
\end{corollary}

The quantity $\xi_M$ allows us to compare different sets of graphs.  If one set of graphs has a higher $\xi_M$ than another set, then this suggests that the first set of graphs can synchronize the system with less coupling (as expressed by $D_k$).

For time-invariant $G_k$ (or $G_k(t)$ belongs to a finite fixed set of matrices for all $t$), we can use a bisection approach to compute $\xi_M$ \cite{wu:synch-inequality:2006}.  
Let $Q$ be an $n$ by $n-1$ matrix whose columns form an orthonormal basis of $e^{\perp}$.
Consider for a fixed 
$\xi$ the following feasibility semidefinite programming (SDP) problem:

{\em Find $U=U^T$ such that $\forall k$ $U(G_k-\xi I)\succeq 0$, $Ue = 0$
$ \forall i\neq j$ $U_{ij}\leq 0$ and $Q^TUQ\succeq I$.}

Using a bisection search to repeatedly solve the SDP problem above allows us to compute $\xi_M$.

\begin{algorithm} 
\SetAlgoLined
\KwData{$G_k$, $\epsilon$, lower and upper bound $lb$, $ub$}
\KwResult{$\xi_M$}
$\xi \leftarrow \frac{1}{2}(lb+ub)$\;
\While{$|ub-lb| > \epsilon$}{ 
      \eIf {SDP problem is infeasible for $\xi$}{
	       $ub \leftarrow \xi$\;
       }{ 
            $lb \leftarrow \xi$\;
            }
     $\xi \leftarrow \frac{1}{2}(lb+ub)$\;
     }
$\xi_M \leftarrow \xi$\;   
\caption{Compute $\xi_M$}\label{alg:xi}
\end{algorithm}

For example, consider the following set of Laplacian matrices $G_k$
\begin{equation*} 
G_1 = \begin{pmatrix}
    1  &   0 &   -1  &   0 &    0\\
     0  &   2  &   0  &  -1  &  -1\\
    -1   &  0  &   3  &  -1  &  -1\\
    -1  &  -1  &   0  &   3  &  -1\\
     0  &   0  &   0  &  -1   &  1
\end{pmatrix}
\end{equation*}

\begin{equation*}
G_2 = \begin{pmatrix}
     2  &  -1  &   0  &   0  &  -1\\
     0  &   1  &   0  &  -1  &   0\\
     0  &   0  &   2  &  -1  &  -1\\
    -1  &   0  &   0  &   1  &   0\\
    -1  &   0  &  -1  &  -1  &   3    
\end{pmatrix}
\end{equation*}

\begin{equation*}
G_3 = \begin{pmatrix}
     1  &  -1  &   0  &   0  &   0\\
    -1  &   2  &  -1  &   0  &   0\\
     0  &  -1  &   1  &   0  &   0\\
     0  &   0  &  -1  &   1  &   0\\
    -1  &  -1  &  -1  &   0  &   3
\end{pmatrix}
\end{equation*}
corresponding to the directed graphs in Fig. \ref{fig:graphs}.

\begin{figure}[htbp]
\centerline{\includegraphics[width=1.2in]{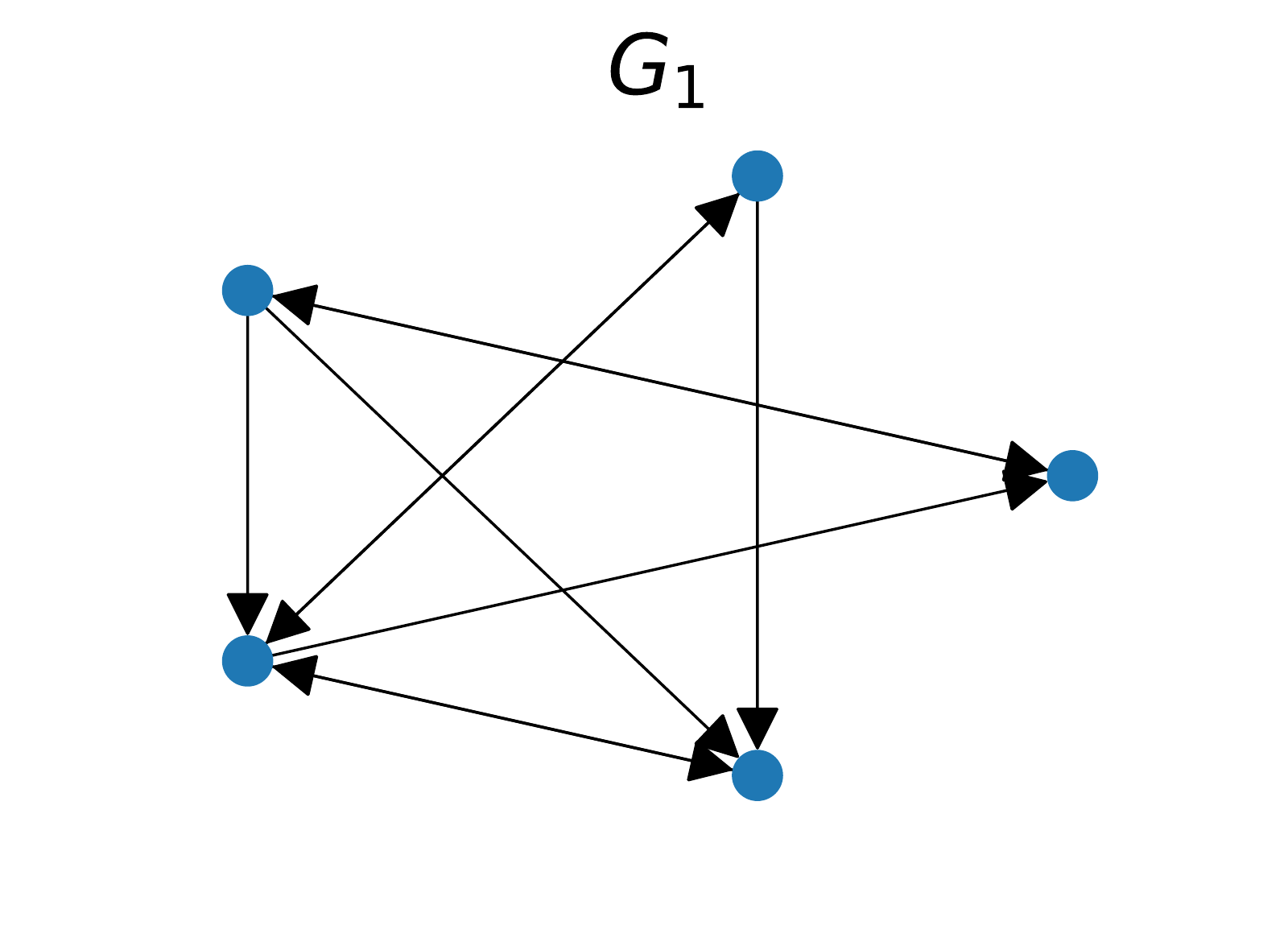}\includegraphics[width=1.2in]{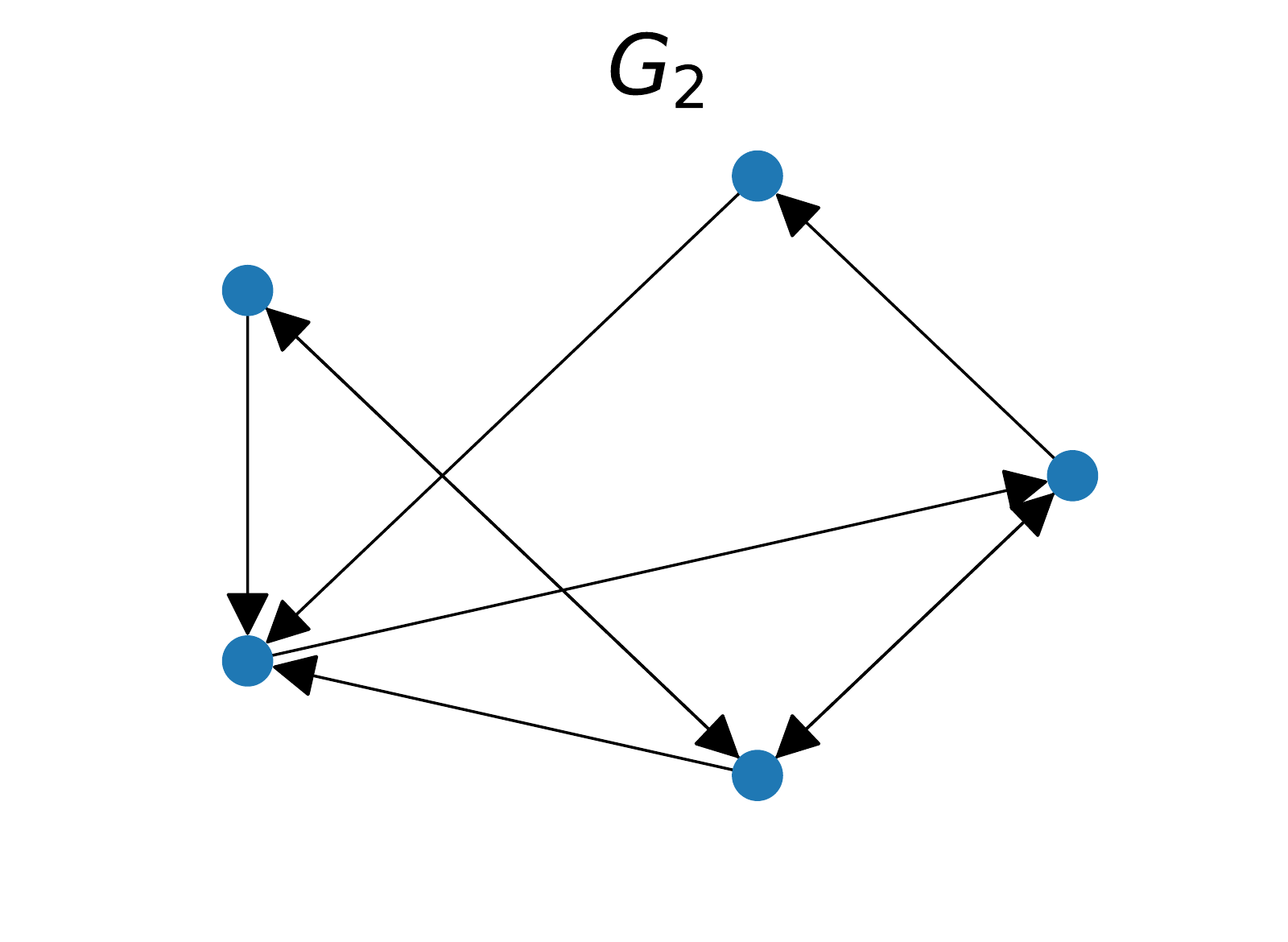}\includegraphics[width=1.2in]{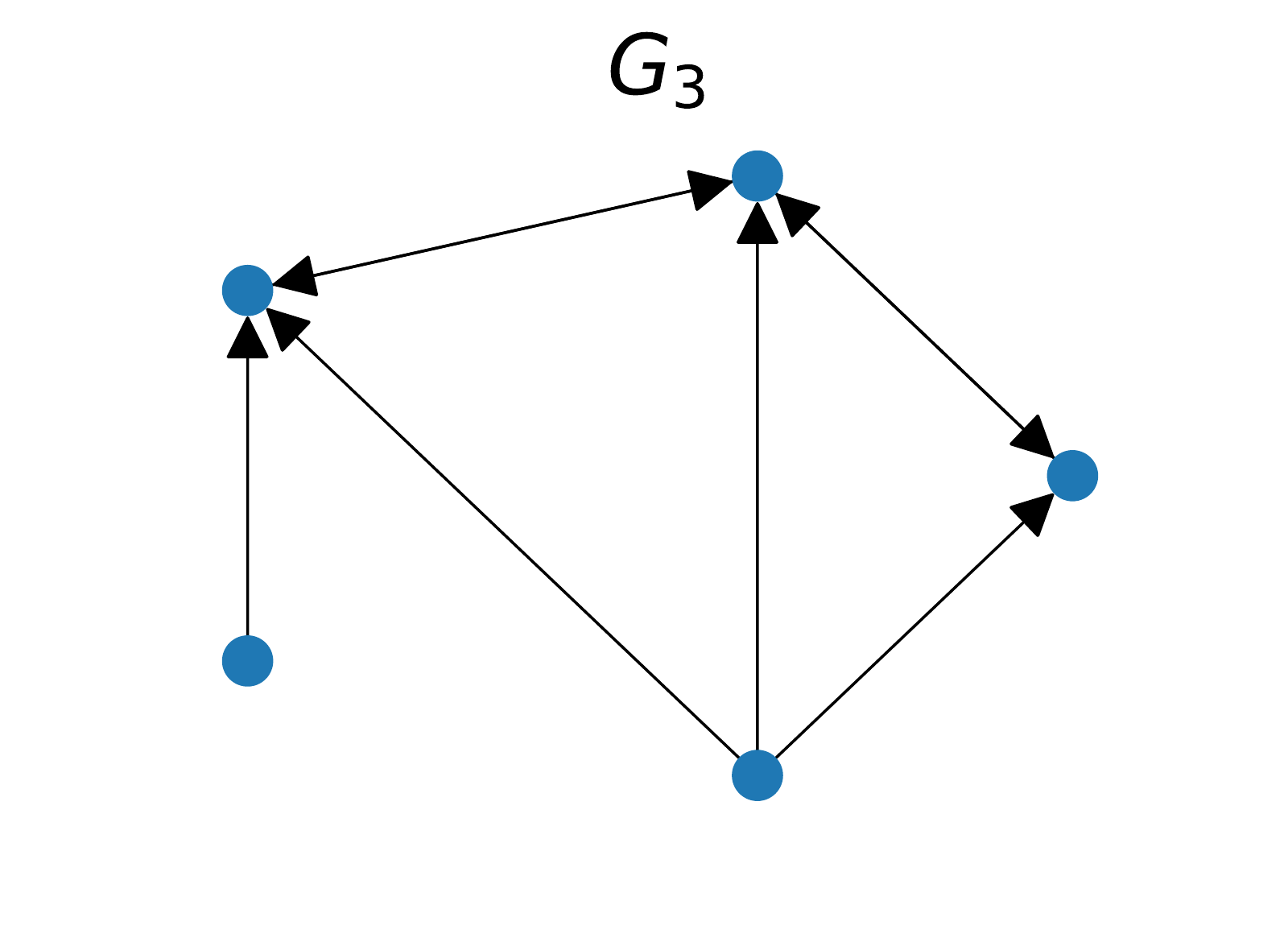}}
\caption{Directed graphs with Laplacian matrices $G_1$, $G_2$, and $G_3$.}
\label{fig:graphs}
\end{figure}

Running Algorithm \ref{alg:xi} using the software package CVX \cite{cvx,gb08} to solve the SDP problem, we find that $\xi_M = 0.838$.
For real zero sums matrices $G_k$, the value of $\xi_M$ is finite and an upper and lower bound can be explicitly computed (which are needed to initialize the bisection search in Algorithm \ref{alg:xi}).
\begin{theorem}\label{thm:lowerbound}
If all $G_k$ are all real zero row sums matrices, then $\xi_M$ exists and is lower bounded by
$ \xi_M \geq \min_k \min_{x\neq 0,x\perp e}\frac{x^TG_k x}{x^Tx} $.
\end{theorem}
\begin{proof}
The proof of Theorem 2 in \cite{wu:synch-inequality:2006} shows that for $U$ equal to the Laplacian matrix of the complete graph,
$U(G-\mu I)\succeq 0$, where 
$\mu = \min_{x\neq 0,x\perp e} x^TGx$ and the conclusion follows. \eop
\end{proof}

\begin{theorem} \label{thm:ximbound}
If all $G_k$ are real matrices with nonpositive off-diagonal elements and zero row and column sums, then $\xi_M \geq 0$.  If in addition $G_k+G_k^T$ are all irreducible, then
$\xi_M > 0$.
\end{theorem}
\begin{proof}
Follows from Corollary 2 in \cite{wu:synch-inequality:2006} and Theorem \ref{thm:lowerbound}.\eop
\end{proof}

Note that Theorem \ref{thm:ximbound} implies that if the directed graphs whose Laplacian matrices are $G_k$ are balanced\footnote{A directed graph is balanced if the indegree of each vertex is equal to its outdegree.}, then $\xi_M\geq 0$. If they are balanced and strongly connected\footnote{Balanced graphs are weakly connected if and only if they are strongly connected.}, then $\xi_M > 0$.

\begin{definition}
For a real matrix $G$ with zero row sums, define $\mu_2(G)$ as the minimum of $Re(\lambda)$ among all eigenvalues $\lambda$ not corresponding to the eigenvector $e$.
\end{definition}

\begin{theorem}
If $G_k$ are all real matrices with zero row sums, then $\xi_M \leq \min_k \mu_2(G_k)$.
\end{theorem}
\begin{proof}
Similar to the proof of Theorem 3 in \cite{wu:nonreciprocal:2003}, 
let $\lambda_k$ be an eigenvalue of $G_k$ with eigenvector $v_k\neq e$. Let $U\in {\cal W}$ be such that
$U(G_k - \mu I) \succeq 0$ for all $k$ for some real number $\mu$. Since the kernel of $U$ is spanned by $e$, $v_k$ is not in the kernel of $U$. Since 
$(G_k - \mu I)v_k = \left(\lambda_K - \mu\right)v_k$, this
implies that $v_k^*U(G_k-\mu I)v = \left(\lambda_K - \mu\right)v_k^*Uv_k$. Since $U(G_k-\mu I)\succeq 0$ this
means that $Re(v_k^*U(G_k-\mu I)v_k)\geq 0$.  Since $U$ is symmetric positive semidefinite
and $v_k$ is not in the kernel of $U$, $v_k^*Uv > 0$. This means that $Re(\lambda_k) - \mu \geq 0$ for all $k$.\eop
\end{proof}

It is clear that $\xi_M(\{G_1,\cdots G_r\}) \leq \min_i \{\xi_M(G_1),\cdots ,\xi_M(G_r)\}$.
For the 3 graphs in the above example, 
$0.838 = \xi_M(\{G_1,\cdots G_r\}) < \min_i \{\xi_M(G_1),\cdots ,\xi_M(G_r)\} = 0.852$.
One interesting question to ask is what are the conditions for which $\xi_M(\{G_1,\cdots G_r\}) = \min_i \{\xi_M(G_1),\cdots ,\xi_M(G_r)\}$.The following result gives one such condition.

\begin{theorem} \label{thm:normal}
If $G_k$ are all real normal matrices with zero row sums, then $\xi_M = \min_k \min_{x\neq 0,x\perp e}\frac{x^TG_k x}{x^Tx} = \min_k \mu_2(G_k)$.
\end{theorem}
\begin{proof}
Follows from the fact that for a real normal zero row sum matrix $G$, $\xi_M(G) = \min_{x\neq 0,x\perp e} x^TGx = \mu_2(G)$ with a fixed $U=L_K\in {\cal W}$ \cite{wu:synch-inequality:2006} and thus $\xi_M \geq \min_k \mu_2(G_k)$.\eop
\end{proof}

\subsection{Coupled linear systems}\label{sec:linear}
Consider now the case where $f(x,t) = Ax+b$ for some matrix $A$ and vector $b$ and the coupling does not vary with time.  In this case, the state equation is described by
$\dot{x} = (I\otimes A)x + \sum_{k=1}^r(G_k\otimes D_k)x + (I\otimes b)$.
For the case $r=1$, the eigenvalues of $I\otimes A + G\otimes D$ are equal to $A+\lambda_{j}D$ where $\lambda_j$ are the eigenvalues of $G$.
If the matrices $G_k$ can be simultaneously  transformed to upper triangular form, then by Lemma \ref{lem:kron-ev}
the eigenvalues of $(I\otimes A)+\sum_{k=1}^r (G_k\otimes D_k)$ are equal to the eigenvalues of $A+\sum_k\lambda_{kj}D_k$, where $\lambda_{kj}$ are the eigenvalues of $G_k$.
Thus we can study the stability of this linear system by looking at the eigenvalues of $A+\sum_k\lambda_{kj}D_k$ for each set of eigenvalues $\{\lambda_{1j},\cdots , \lambda_{rj}\}$.
Extending the approach in \cite{wu:chua_kronecker_product_synch_95}, we can study a map $\phi :\C^r \rightarrow \R$ which maps $\{\lambda_{1j},\cdots , \lambda_{rj}\}$ to the most positive real part of the eigenvalues of $A+\sum_k\lambda_{kj}D_k$.  
Define $\S = \{x:\phi(x) < 0\}$.  Then the network is synchronizing if all sets of eigenvalues $\{\lambda_{1j},\cdots, \lambda_{rj}\}$ lie in $\S$, except for the set $\{0,\cdots,0\}$ corresponding to the eigenvectors $e$ for each $G_k$. The shape and size of $\S$ provide insight on how the graph topologies affect the synchronizability of the networked system.

\section{Discrete-time systems}\label{sec:discrete}

Consider a network of coupled discrete-time systems with the following state equations:
\begin{equation}\begin{array}{lcl} \label{eqn:discrete-time}
&&x(p+1)  =  \left(\begin{array}{c}x_1(p+1)\\\vdots \\x_n(p+1)\end{array}\right) \\
&&= \left(I - \sum_kG_k(p)\otimes D_k(p)\right)
\left(\begin{array}{c}f(x_1(p),p)\\\vdots \\f(x_n(p),p)\end{array}\right) + u(p) \\
&&= \left(I-\sum_kG_k(p)\otimes D_k(p)\right) F(x(p),p) + u(p)
\end{array}
\end{equation}

\begin{theorem} \label{thm:cml_top}
Consider the network of coupled discrete-time systems with state equation Eq. (\ref{eqn:discrete-time}) where
$G_k(p)$ are normal commuting $n$ by $n$ matrices with zero row sums for each $k$ and $p$.  Let
$V$ be a symmetric positive definite matrix such that
$\left(f(z,p)-f(\tilde{z},p)\right)^TV\left(f(z,p)-f(\tilde{z},p)\right) \leq c(z-\tilde{z})^TV(z-\tilde{z})$
for some $c>0$ and all $p$, $z$, $\tilde{z}$.  Let $|u_l(p)-u_m(p)| \rightarrow 0$ for all $l$, $m$ as $p\rightarrow\infty$ and let $V$ be decomposed\footnote{For instance via the Cholesky decomposition.} as
$V=C^TC$. Then the coupled systems synchronize, i.e. $x_l(p)\rightarrow x_m(p)$ for all $l$,$m$ as $p\rightarrow\infty$ if for each $p$ and for each $\lambda_{ki} \in S_k$
\begin{equation}\label{eqn:discrete}
\left\|I-\sum_k\lambda_{ki}CD_k(p)C^{-1}\right\|_2< \frac{1}{\sqrt{c}}
\end{equation} 
where $S_k$ are the eigenvalues of $G_k$ not corresponding to $e$.
\end{theorem}

\begin{proof} 
First consider the case $V=I$, i.e., $f(x,p)$ is Lipschitz 
continuous in $x$ with Lipschitz
constant $\sqrt{c}$. Let $W = I-\sum_k G_k\otimes D_k$.
Since $G_k(t)$ is normal and simultaneously diagonalizable, $G_k(t)$ has an orthonormal set of 
eigenvectors $v_i$ with $e=v_1$.
Denote $\cal A$ as the subspace of vectors of the form
$e\otimes v$.  
Since $G_ke = 0$
it follows that $\A$ is in the kernel of $G_k\otimes D_k$.  
Let us denote the subspace orthogonal to $\A$ by $\B$.  A vector $x(p)$ is decomposed as $x(p) =
y(p)+z(p)$ where $y(p)\in \A$ and $z(p)\in \B$.  By hypothesis
$\|F(x(p),p)-F(y(p),p)\|\leq c\|z(p)\|$.  We can decompose
$F(x(p),p)-F(y(p),p)$ as $a(p)+b(p)$ where
$a(p)\in \A$ and $b(p)\in \B$.  Note that $F(y(p),p)$ and $a(p)$ are in
$\A$ and therefore $(G_k\otimes D_k)F(y(p),p) =  (G_k\otimes D_k)a(p) = 0$.
\begin{eqnarray*}& &x(p+1) =WF_k(x(p),p) +u(p)\\ 
&=& 
a(p) + F_k(y(p),p) + W b(p) +u(p)
\end{eqnarray*}
Since $a(p)\bot b(p)$, $\|b(p)\| \leq
\|F(x(p),p)-F(y(p),p)\| \leq c\|z(p)\|$.  Since $b(p)\in \B$,
it is of the form $b(p) = \sum_{i>1,j} \alpha_{i} v_{i}\otimes f_{j}(p)$.
\begin{eqnarray*}&&W b(p) = 
\sum_{i>1,j} \alpha_{ij} v_i \otimes f_j - \sum_{i>1,j,k} \alpha_{ij}\lambda_{ki} v_i\otimes D_k f_j \\&&= 
\sum_{i>1,j} \alpha_{ij} v_i\otimes \left(I-\sum_k \lambda_{ki}D_k\right)f_j
\end{eqnarray*}
By hypothesis this implies that
\begin{eqnarray*}&&\left\|W b(p)\right\|^2 = \sum_{i>1,j} |\alpha_{ij}|^2 \left\|\left(I-\sum_k \lambda_{ki}D_k\right) f_j\right\|^2\\
&<& \sum_{i>1,j} |\alpha_{ij}|^2\|f_j\|^2/c \leq \|z(p)\|^2
\end{eqnarray*}  
Let $w(p)$ be the orthogonal projection of $u(p)$ onto $\B$.
Since $a_k(p)+F_k(y(p),p)$ is in $\A$, $z(p+1)$ is the orthogonal
projection of $W b(p) + u(p)$ onto $\B$ and thus $\|z(p+1)\| \leq
\|W b(p)\|+\|w(p)\|$.  The above shows that there exists $\beta < 1$ such that
$\|W b(p)\| \leq \beta\|z(p)\|$.
Thus $\|z(p+1)\| \leq \beta\|z(p)\| + \|w(p)\|$.  Note that $\sum_k (G_k\otimes D_k)u(p) \rightarrow 0$ as $p\rightarrow\infty$ and thus $w(p) \rightarrow 0$ as $p\rightarrow \infty$.  This implies that  $z(p)\rightarrow 0$ as $p\rightarrow
\infty$ and $x(p)$ approaches the synchronization 
manifold $\A$.

For general $V$, let $C^TC$ be the decomposition of $V$.
Applying the state transformation $\tilde{x} = (I\otimes C)x$, we get
$\tilde{x}(p+1) = \sum_k(I-G_k\otimes CD_kC^{-1})\left(\begin{array}{c} \tilde{f}_k(\tilde{x}_1(p),p)\\ \vdots \\
\tilde{f}_k(\tilde{x}_n(p),p)\end{array}\right)+\tilde{u}(p)$
where $\tilde{f}(\tilde{x}_i,p) = Cf(C^{-1}\tilde{x}_i,p)$.
We can then apply the result from the $V=I$ case
and the conclusion follows.\eop
\end{proof}

Theorem \ref{thm:cml_top} implies that if $G_k = \xi L_k$, $L_k$ are commuting normal Laplacian matrices with the graph sum's reversal containing a spanning directed tree and $\xi$ is small enough, then the discrete system synchronizes.
If $V=I$ and $D_k(p)$ are also normal commuting matrices, then Eq. (\ref{eqn:discrete}) can be simplified to:
$\left|1-\sum_k\lambda_{ki}\mu_{kj}\right|\leq \frac{1}{\sqrt{c}}$
for all $\lambda_{ki}\in S_k$ where $\mu_{kj}$ are the eigenvalues of $D_k$.

\section{Conclusions}
We study coupled dynamical systems coupled via multiple directed networks and provide criteria under which they synchronize for both continuous-time systems and discrete-time systems. In general, under suitable conditions, two conclusions can be drawn from these synchronization criteria.
First, if the coupling between individual systems (as expressed by $D_k$) are strong enough, then the graph sum of all the networks jointly containing a spanning directed tree is sufficient to synchronize the network, even though a single network may not be connected. This is an intuitive conclusion as it implies that there is a node that influences every other node via a set of edges where each edge in the set belongs to at least one of the networks.
Secondly, even if the coupling of a single network $D_k$ is not sufficient to synchronize the systems, but the sum $\sum_k D_k$ does, then the synchronization can still occur with multiple networks if each of the network is well connected, for instance, if they are all strongly connected and balanced.

\label{lastpage}
\bibliography{quant,markov,consensus,secure,synch,misc,stability,cml,algebraic_graph,graph_theory,control,optimization,adaptive,top_conjugacy,ckt_theory,math,number_theory,matrices,power}

\end{document}